\documentclass[journal,draftcls, onecolumn,11pt]{IEEEtran}

\usepackage{amsmath, amssymb,amsthm}
\usepackage{inputenc, graphicx, verbatim, enumerate}
\usepackage{float}
\usepackage{tikz,pgfplots}
\usepackage{url}
\usepackage{xcolor,color}
\usepackage{cite}

\newtheorem{theorem}{Theorem}
\newtheorem{proposition}{Proposition}
\newtheorem{lemma}{Lemma}

\newtheorem{definition}{Definition}

\newtheorem{conjecture}{Conjecture}

\newcommand{\F}{\mathbb{F}}

\newcommand{\N}{\mathbb{N}}

\newcommand{\FF}{\mathcal{F}}

\newcommand{\CC}{\mathcal{C}}

\newcommand{\cl}{\mathrm{cl}}

\begin{document}
\title{Uniform Minors in Maximally Recoverable Codes}

\author{%
  \IEEEauthorblockN{Matthias Grezet, Thomas Westerb\"{a}ck, Ragnar Freij-Hollanti, and Camilla Hollanti}
  \thanks{M. Grezet, R. Freij-Hollanti, and C. Hollanti are with the Department of Mathematics and Systems Analysis, Aalto University, FI-00076 Aalto, Finland (email: $\{$firstname.lastname, ragnar.freij$\}$ @aalto.fi).}
  \thanks{T. Westerb\"{a}ck is with the Division of Applied Mathematics, UKK, M\"{a}lardalen University, H\"{o}gskoleplan 1, Box 883, 721 23 V\"{a}ster\r{a}s, Sweden (e-mail: thomas.westerback@mdh.se).}
  \thanks{This work was supported in part by the Academy of Finland, under grants 276031, 282938, and 303819, and by the Technical University of Munich -- Institute for Advanced Study, funded by the German Excellence Initiative and the EU 7th Framework Programme under grant agreement 291763, via a Hans Fischer Fellowship.}
   \thanks{Copyright \textcopyright \, 2019 IEEE. Personal use of this material is permitted.  However, permission to use this material for any other purposes must be obtained from the IEEE by sending a request to pubs-permissions@ieee.org.}
  }

\maketitle

\begin{abstract}
In this letter, locally recoverable codes with maximal recoverability are studied with a focus on identifying the MDS codes resulting from puncturing and shortening. By using matroid theory and the relation between MDS codes and uniform minors, the list of all the possible uniform minors is derived. This list is used to improve the known non-asymptotic lower bound on the required field size of a maximally recoverable code. 
\end{abstract}

%
\IEEEpeerreviewmaketitle





\section{Introduction}

With the exponential growth of data needed to be stored remotely, distributed storage systems (DSSs) using erasure-correcting codes have become attractive due to their high reliability and low storage overhead. A class of codes called \emph{locally recoverable codes} (LRCs) has been introduced in \cite{gopalan12, papailiopoulos12} as an alternative to traditional \emph{maximum distance separable} (MDS) codes to improve node repair efficiency by allowing one failed node to be repaired by only accessing a few other nodes. 

A linear $(n,k,r)$-LRC is a linear code of length $n$ and dimension $k$ over $\F_q$ such that every codeword symbol $i \in [n]=\{1, \ldots , n\}$ is contained in a \emph{repair set} $R_{i} \subseteq [n]$ with $|R_{i}|\leq r+1$ and the minimum Hamming distance of the restriction of the code to $R_{i}$ is at least $2$. In other words, any symbol can be determined by the values of at most $r$ other symbols. Constructions of distance-optimal LRCs with field size of order $n$ have been given in \cite{tamo14}. 

LRCs with \emph{maximal recoverability} (MR-LRCs) or maximally recoverable codes (also known as partial MDS or PMDS) have been introduced in \cite{chen07}. MR-LRCs are a subclass of distance-optimal LRCs that can correct any erasure pattern that is information-theoretically correctable. Formally, an $(n,k,r)$ MR-LRC is an $(n,k,r)$-LRC whose codeword symbols are partitioned into $g:=n/(r+1)$ disjoint repair sets $R_{i}$ and any set $S$ of $k$ symbols with $R_{i} \nsubseteq S$ is an information set. The number of heavy (global) parity checks is $h:=n-k-g=gr-k$. This definition can be extended to allow the repair sets to correct $\delta-1$ erasures but for the clarity of this letter, we will only consider $\delta=2$. 

MR-LRCs drew a lot of attention recently with many papers being devoted to the construction of general classes of MR-LRCs over the lowest possible field size. While a field size linear in $n$ is sufficient for optimal LRCs, known constructions of MR-LRCs for any parameters $(n,k,r,\delta)$ are generally exponential in $r$ or $h$. A general construction for $\delta=2$ local erasures with field size of order $k^{h}$ was obtained in \cite{gopalan14}. The best constructions so far for MR-LRCs tolerating $\delta-1$ local erasures were given in \cite{gabrys18,martinez18}, where \cite{gabrys18} obtained field sizes of order $ (r+\delta-1)n^{(\delta+1)h-1}$ and $\max\{ g, (r+\delta-1)^{\delta+h} \}^{h}$, and \cite{martinez18} obtained a field size of order $(g+1)^{r}$. 

However, little is known regarding the lower bound on the required field size $q$. In \cite{gopalan14}, the authors proved that by puncturing one element per repair set, the resulting code is an $[n-g,k,n-g-k+1]$ MDS code and therefore $q \geq k+1$. Recently, \cite{gopi19} gave the first asymptotic superlinear lower bound for MR-LRCs tolerating $\delta-1$ erasures when $h$ is constant and $r$ may grow with $n$. The bound is the following:
\[
q \geq \Omega(n r^{\alpha}) \; \text{where} \; \alpha = \frac{\min \{ \delta-1, h-2\lceil h/g \rceil \} }{\lceil h/g \rceil }.
\]

In this letter, we pursue the approach started by \cite{gopalan14} and identify, for each dimension, the largest length of an MDS code obtained by puncturing and shortening. Our main tools to achieve this come from matroid theory. The link between MR-LRCs and matroids was already used in \cite{lalitha15} where the authors computed the Tutte polynomial of MR-LRCs to derive the weight enumerator and higher support weights. Here, we work with the collection of flats and matroid minors to construct the largest possible uniform minors in MR-LRCs and thus, the largest MDS codes. These minors are then used to improve the non-asymptotic lower bound found in \cite{gopalan14}, both with and without assuming the MDS conjecture.



\section{Preliminaries}
\label{section:preliminaries}

We denote the set $\{1,2,\ldots, n\}$ by $[n]$ and the set of all subsets of $[n]$ by $2^{[n]}$. A generator matrix of a linear code $\CC$ is $G_{\CC} = (\mathbf{g_{1}} \cdots \mathbf{g_{n}})$ where $\mathbf{g_{i}} \in \F_{q}^{k}$ is a column vector for $i \in [n]$. Matroids have many equivalent definitions in the literature. Here, we choose to define matroids via their rank functions. Much of the contents in this section can be found in more detail in \cite{freij18}.

\begin{definition}
\label{def:matroid_rank}
A \emph{(finite) matroid} $M=(E, \rho)$ is a finite set $E$ together with a \emph{rank function} $\rho:2^E \rightarrow \mathbb{Z}$ such that for all subsets $X,Y \subseteq E$,
\[
\begin{array}{rl} 
(R.1) & 0 \leq \rho(X) \leq |X|,\\
(R.2) & X \subseteq Y \quad \Rightarrow \quad \rho(X) \leq \rho(Y),\\
(R.3) & \rho(X) + \rho(Y) \geq \rho(X \cup Y) + \rho(X \cap Y). 
\end{array}
\]
\end{definition}

There is a unique matroid $M_{\CC}$ associated to a linear code $\CC$ where $E=[n]$ and $\rho(X)$ is the dimension of the restriction of $\CC$ to $X$ for $X\subseteq [n]$.

Two matroids $M_1 = (E_1, \rho_1)$ and $M_2 = (E_2, \rho_2)$ are \emph{isomorphic} if there exists a bijection $\psi: E_1 \rightarrow E_2$ such that $\rho_2(\psi(X)) = \rho_1(X)$ for all subsets $X \subseteq E_1$. We denote two isomorphic matroids by $M_{1} \cong M_{2}$. 

Let $M = (E,\rho)$ be a matroid. The \emph{closure} operator $\cl:~2^E \rightarrow 2^E$ is defined by $ \cl(X) = \{e \in E  : \rho(X \cup e) =\rho(X)\}$. A subset $F \subseteq E$ is a \emph{flat} if $\cl(F) = F$ and the collection of flats is denoted by $\FF(M)$. 


\begin{definition}
The \emph{uniform matroid} $U_n^k=([n],\rho)$ is a matroid with a ground set $[n]$ and a rank function $\rho(X)=~\min\{|X|, k\}$ for $X\subseteq [n]$. In particular, the flats are $\FF(M)=\{ F \subseteq [n] : |F|<k \} \cup [n]$. 
\end{definition}

The following straightforward observation gives a characterization of MDS codes. 

\begin{proposition}
\label{prop:mds_unif}
A linear code $\CC$ is an $[n,k]$-MDS code of length $n$ and dimension $k$ if and only if $M_{\CC}$ is the uniform matroid $U_n^k$.
\end{proposition}

There are several elementary operations that are useful for explicit constructions of matroids, as well as for analyzing their structure. 

\begin{definition}
Let $M = (E,\rho)$ be a matroid and $X, Y \subseteq E$. Then
\begin{enumerate}
\item The \emph{restriction} of $M$ to $Y$ is the matroid $M|Y = (Y, \rho_{|Y})$, where $\rho_{|Y}(A) = \rho(A)$ for $A \subseteq Y$.
\item The \emph{contraction} of $M$ by $X$ is the matroid $M/X = (E-X, \rho_{/X})$, where $\rho_{/X}(A) = \rho(A \cup X) - \rho(X)$ for $A \subseteq E-X$.
\item For $X \subseteq Y$, a \emph{minor} of $M$ is the matroid $M|Y/X = (Y-X, \rho_{|Y/X})$ obtained from $M$ by restriction to $Y$ and contraction by $X$. Observe that this does not depend on the order in which the restriction and contraction are performed.
\end{enumerate}
\end{definition}

The deletion of $M$ by $Y$, denoted by $M \setminus Y$, is the restriction of $M$ to $E-Y$. These operations can be equivalently defined via the generator matrix $G_{\CC}$ of a code $\CC$ of length $n$. If we label the columns of $G_{\CC}$ from $1$ to $n$, then the restriction to $Y \subseteq [n]$ is the same as considering the submatrix formed by the columns indexed in $Y$ and the contraction by $X\subseteq [n]$ is the projection from the columns indexed in $X$. Thus, the deletion and contraction correspond to puncturing and shortening of codes, respectively. We can also describe the flats of a minor.

\begin{proposition}
\label{prop:flats_minor}
Let $M = (E,\rho)$ be a matroid and $X,F \subseteq E$ with $F \in \FF(M)$, then
\begin{enumerate}
\item $\FF(M/F)=\{A \subseteq E-F : A \cup F \in \FF(M) \}$,
\item $\FF(M\setminus X) = \{ F - X : F \in \FF(M) \}$.
\end{enumerate}
\end{proposition}



\section{Uniform minors and a lower bound on $q$}
\label{section:section3}

As mentioned in the introduction, this letter pursues two objectives: classifying the uniform minors or MDS codes inside an MR-LRC and improving the lower bound on the required field size. The second problem is highly related to the MDS conjecture.

\begin{conjecture}[\hspace{1sp}\cite{segre55}]
If $k \leq q$ then a linear $[n,k]$-MDS code over $\F_{q}$ has length $n \leq q+1$ unless $q=2^{m}$ and $k=3$ or $k=q-1$, in which case $n \leq q+2$. 
\end{conjecture}

The conjecture is proven when $q$ is a prime or when $k\leq 2p-2$ for $q=p^m$ in \cite{ball12}. Without assuming the MDS conjecture, the following lemma bounds the field size.

\begin{lemma}[\hspace{1sp}\cite{ball12} Lemma 1.2]
\label{lemma:mds_lowerbound}
Any $[n,k]$-MDS code over $\F_{q}$ satisfies $n \leq q+k-1$. 
\end{lemma}

Regarding the classification of the uniform minors, we first give the structure of the flats of the associated matroid to an MR-LRC. For simplicity, if $M$ is the matroid associated to an $(n,k,r)$ MR-LRC, then $M$ is called an \emph{$(n,k,r)$-MR matroid}. 

\begin{proposition}
Let $M=(E,\rho)$ be an $(n,k,r)$-MR matroid. Then the flats are
\begin{align*}
\FF(M) = \{ F & \subseteq E : \text{for all } i \in [g], R_{i} \subseteq F \text{  or } |R_{i} \cap F| \leq \\ 
	& r-1, |F|-|\{ i : R_{i} \subseteq F \}|<k \} \cup E.
\end{align*}
The rank of a flat $F \neq E$ is $\rho(F)=|F|-|\{ i : R_{i} \subseteq F \}|$. 
\end{proposition}

\begin{proof}
A set $A$ with $\rho(A)<k$ is not a flat if and only if there exists a repair set $R_{i}$ such that $|R_{i} \cap A|=r$. Indeed, if $\{e\}=R_{i}-A$ then $e\in \cl(R_{i}-A)=R_{i}$ and therefore $e\in \cl(A)-A$. Moreover, the rank function of $M$ is given by $\rho(A)=\min\{ k,|A|-\{i : R_{i} \subseteq A \}| \}$. 
\end{proof}

The following theorem by Higgs is known as the \emph{Scum Theorem}. It significantly restricts the sets $A\subseteq B\subseteq E$ that one must consider in order to find all minors of $M$ as $M|B/A$. 

\begin{theorem}[\hspace{1sp}\cite{crapo70} Proposition 3.3.7] 
\label{thm:scum}
Let $N=(E_N,\rho_N)$ be a minor of a matroid $M=(E,\rho)$. Then there is a pair of sets $A \subseteq B\subseteq E$ with $\rho(A) = \rho(E)-\rho_N(E_N)$ and $\rho(B)=\rho(E)$, such that $M|B/A \cong N$. Further, if for all $e \in E_N$, we have $\rho_N(e)>0$, then $A$ can be chosen to be a flat of $M$. 
\end{theorem}



The next four propositions classify all the uniform minors in an MR-LRC. One uniform minor has already been obtained in \cite{gopalan14} by deleting one element per repair set. It can be formulated as follows.

\begin{proposition}[\hspace{1sp}\cite{gopalan14} Theorem 19]
\label{prop:uniform_k_nm}
Let $M$ be an $(n,k,r)$-MR matroid. Then $M$ contains a $U_{n'}^{k}$ minor where
\begin{equation}
\label{eq:uniform_k_nm}
n'=n-g.
\end{equation}
\end{proposition}

\begin{proposition}
\label{prop:uniform_r_np}
Let $M$ be an $(n,k,r)$-MR matroid. Then $M$ contains a $U_{n'}^{r}$ minor where 
\begin{equation}
\label{eq:uniform_r_np}
n'=n-k+r-\left\lceil\frac{k}{r}\right\rceil +1.
\end{equation}
\end{proposition}

\begin{proof}
By Theorem \ref{thm:scum}, we are looking for a flat $F \in \FF(M)$ such that $\rho(F)=k-r$ and for all $A \subseteq E-F$ with $|A|\leq r-1$ we have $F \cup A \in \FF(M)$. Then, by Proposition \ref{prop:flats_minor}, $M/F \cong U_{n-|F|}^{r}$. When $r \neq k$, we will also use an extra deletion.

The second condition implies that if there exists a non-empty $B_{i} \subsetneq R_{i}$ with $B_{i} \subseteq F$, then either $R_{i} \subseteq F$ or we need to delete an element from $M$. The reason is that if $|B_{i}|=r$, then $R_{i} \subseteq F$ since $F$ is a flat. If $|B_{i}|\leq r-1$, then let $e \in R_{i}-B_{i}$ and choose $A = R_{i}-B_{i}-\{e\}$. We have $|A|\leq r-1$ but $F \cup A$ is not a flat because $e \in \cl(F\cup A) - F\cup A$. 

Let us first assume that $r \mid k$. Because of the previous argument, we need $F=\bigcup_{i\in \left\lbrace 1, \ldots , \frac{k}{r}-1 \right\rbrace } R_{i}$. Then, for all $A \subseteq E-F$ with $|A|\leq r-1$, we have $F \cup A \in \FF(M)$. Thus, $M/F$ is a uniform minor with rank $r$ and size 
\[
n'=n-|F| = n-\left(\frac{k}{r}-1\right)(r+1)= n-k+r-\frac{k}{r}+1. 
\]

Assume now that $r \nmid k$. In this case, we need to add a part of a repair set to complete the rank and delete an element to remove the unwanted flat. Let $F_{1}=\bigcup_{i \in \left\lbrace 1, \ldots , \left\lfloor \frac{k}{r} \right\rfloor -1 \right\rbrace} R_{i}$. Then, we have
\[
k-2r< \rho(F_{1})=\left( \left\lfloor \frac{k}{r} \right\rfloor -1\right)r < k-r.
\]
Since the rank of the union of all repair sets is $k$, there exists an extra repair set $R_{\left\lfloor \frac{k}{r} \right\rfloor}$. Let $B \subseteq R_{\left\lfloor \frac{k}{r} \right\rfloor}$ such that $|B|=k-r-\rho(F_{1})$. Notice that $0 < |B| <r$. Now let $F=F_{1} \cup B$. Then, $\rho(F)=k-r$ and
\[
|F|=\left( \left\lfloor \frac{k}{r} \right\rfloor -1\right)(r+1) + |B| = k - r + \left\lceil \frac{k}{r} \right\rceil -2.
\]

Furthermore, for all $A \subseteq E-F-R_{\left\lfloor \frac{k}{r} \right\rfloor}$, we have $F \cup A \in \FF(M)$. It remains to delete one element in $R_{\left\lfloor \frac{k}{r} \right\rfloor}$ to get rid of the flat $R_{\left\lfloor \frac{k}{r} \right\rfloor}-B$. To this end, let $e \in R_{\left\lfloor \frac{k}{r} \right\rfloor} -B$. We have $\rho(M/F \setminus \{e\})=r$ and $\FF(M/F\setminus \{e\}) = \{ A \subseteq E -F - \{e\} : |A|\leq r-1 \} \cup (E-F-\{e\})$. Hence $M/F\setminus \{e\}$ is a uniform minor $U_{n'}^{r}$ where 
\[
n' = n-|F|-1 = n -k +r - \left\lceil \frac{k}{r} \right\rceil +1.
\] 
\end{proof}

\begin{proposition}
\label{prop:uniform_kp_np}
Let $M$ be an $(n,k,r)$-MR matroid and $2\leq k'\leq r-1$. Then $M$ contains a $U_{n'}^{k'}$ minor where 
\begin{equation}
\label{eq:uniform_kp_np}
n'=n-k+k'-\max\{j,0\} \; \text{with} \; j=\left\lfloor \frac{-h}{k'}\right\rfloor +g.
\end{equation}
\end{proposition}

\begin{proof}
By Theorem \ref{thm:scum}, we are looking for a flat $F \in \FF(M)$ such that $\rho(F)=k-k'$ and for all $A \subseteq E-F$ with $|A|\leq k'-1$ we have $F \cup A \in \FF(M)$.

The second condition implies that if $B_{i} \subsetneq R_{i}$ and $B_{i} \subseteq F$, then $|B_{i}| \leq r-k'$. Otherwise, if $|B_{i}| \geq r+1 - k'$, we have $|R_{i}-B_{i}| \leq k'$. Then, let $e \in R_{i}-B_{i}$ and choose $A=R_{i}-B_{i}-\{e\}$. We have $|A| =r+1-|B_{i}| -1 \leq k'-1$ and $|B_{i} \cup A|=r$. Therefore, $F \cup A$ is not a flat since $e \in \cl(F\cup A) - F \cup A$. To construct $F$, we distinguish two cases depending on the number of repair sets $g=\frac{n}{r+1}$.

Assume first that $i_{1}:=\left\lfloor \frac{k-k'}{r-k'} \right\rfloor < g$. Then, let $F_{1}=\bigcup_{i \in [i_{1}]} B_{i}$ where $B_{i} \subset R_{i}$ with $|B_{i}| = r-k'$. Since $M$ is an MR-matroid, we have that $\rho(F_{1})=|F_{1}|=i_{1} (r-k')$. Let also $X \subset R_{i_{1}+1}$ such that $|X|=k-k'-\rho(F_{1}) = k-k'- \left\lfloor \frac{k-k'}{r-k'} \right\rfloor (r-k') <r-k'$ and define $F=F_{1} \cup X$. 

Hence, $|F|=\rho(F)=k-k'$ and for all $A \subseteq E-F$ with $|A|\leq k'-1$, we have $F \cup A \in \FF(M)$ since $F$ consists of independent elements where no more than $r-k'$ elements of $F$ are contained in the same repair set. 

Assume now that $\left\lfloor \frac{k-k'}{r-k'} \right\rfloor \geq g$. This means that there are not enough repair sets to build an independent set as in the previous case and $F$ has to contain some $R_{i}$. Thus, we are looking for the minimum number $j \in \{1, \ldots, g-1\}$ of repair sets that $F$ has to contain before we can add an independent set. Formally, $j$ is given by
\[
j = \min \left\{ j' \in \N : \left\lfloor \frac{k-k'-j'r}{r-k'} \right\rfloor < g-j' \right\}.
\]

The condition on $j'$ simplifies as follows.
\begin{align*}
\left\lfloor \frac{k-k'-j'r}{r-k'} \right\rfloor < g-j'  & \iff \\
\frac{k-k'-j'r}{r-k'} < g-j' & \iff \\
\frac{k-k'-j'k'}{r-k'} < g & \iff \\
k-k'(j'+1)<(r-k')g & \iff \\
\frac{k-rg}{k'}+g = \frac{-h}{k'} +g < j'+1.
\end{align*}

Therefore, we have that $j=\left\lfloor \frac{-h}{k'}\right\rfloor +g$. Now, let $i_{2}=\left\lfloor \frac{k-k'-jr}{r-k'} \right\rfloor$ and $F_{1}= R \cup B$ where $R=\bigcup_{i \in [j]}R_{i}$ and $B=\bigcup_{i \in \left\lbrace j'+1, \ldots , j'+ i_{2} \right\rbrace } B_{i}$ with $B_{i} \subseteq R_{i}$ such that $|B_{i}|=r-k'$. Then, the rank of $F_{1}$ is $\rho(F_{1})=jr + i_{2} (r-k')$. By definition of $j$, we have $j+i_{2} < g $. Then, let $x=j+i_{2} +1$ and $X \subset R_{x}$ such that $|X|=k-k'-\rho(F_{1})$. Notice that $|X|<r-k'$. Finally, define $F=F_{1} \cup X$. We have indeed that $\rho(F)=k-k'$ and $|F|=j(r+1) + i_{2} (r-k') + |X| = k-k'+j$. Moreover, by the same argument as in the previous case, for all $A \subseteq E - F$ with $|A|\leq k'-1$, we have $F \cup A \in \FF(M)$.  

Hence, $M/F$ is a uniform minor with rank $k'$ and size $n'=n-|F|=n-k+k'-\max\{j,0\}$. 

\end{proof}

We are left with the case $r<k'<k$. In fact, requesting $k'>r$ forces the deletion of one element per repair set and the minor obtained is a subminor of the uniform minor obtained in Proposition \ref{prop:uniform_k_nm}. We state it here for completeness. 

\begin{proposition}
\label{prop:uniform_biggerr_np}
Let $M$ be an $(n,k,r)$-MR matroid and $r<k'<k$. Then $M$ contains a $U_{n'}^{k'}$ minor where 
\begin{equation}
\label{eq:uniform_biggerr_np}
n'=n-g-k+k'.
\end{equation}
\end{proposition}

\begin{proof}
We want $F \in \FF(M)$ and $X \subseteq E-F$ such that $M/F\setminus X \cong U_{n'}^{k'}$. Since $k'>r$, it means that for all $A \subseteq E-(F \cup X)$ with $|A|\leq k'-1$, we have $A \in \FF(M/F \setminus X)$. In particular, sets of size $r$ should also be flats. Therefore, we cannot have $R_{i} \subseteq E-(F \cup X)$ and one element needs to be deleted from $R_{i}$ or be contained in $F$. Since the two options yield the same size $n'$, we can choose to delete them first. Let $X = \bigcup_{i \in [g]} e_{i}$ with $e_{i} \in R_{i}$ and let $M'=M\setminus X$. As in Proposition \ref{prop:uniform_k_nm}, we have $M'\cong U_{n-g}^{k}$. Let $F \subseteq E-X$ with $\rho(F)=|F|=k-k'$. Hence $M\setminus X /F \cong U_{n'}^{k'}$ with $n'=n-g-k+k'$. 
\end{proof}

The techniques developed here easily generalize to the case when $\delta>2$ by taking the size of a repair set to be $r+\delta-1$ and deleting $\delta-1$ elements instead of $1$.


When assuming the MDS conjecture, only the code length matters in the lower bound on the field size. Therefore, assuming the MDS conjecture, the bound on the field size of an MR-LRC is the largest size of a uniform minor minus one except on some special cases when $q$ is even. The next theorem gives the largest size of all the uniform minors found in the previous propositions. As such, it does not depend on the MDS conjecture.

\begin{theorem}
Let $M$ be an $(n,k,r)$-MR matroid with $g=\frac{n}{r+1}$. The largest size of a uniform minor is
\[
n'=\left\lbrace
\begin{array}{ll}
n-\min \left\lbrace g,k-r+\left\lceil \frac{k}{r} \right\rceil -1 \right\rbrace & \text{if } r=2,\\
n-\min\{g,k-r+1\} & \text{if } r\geq 3. \\
\end{array}
\right.
\]
\end{theorem}

\begin{proof}
We compare the sizes obtained in \eqref{eq:uniform_k_nm}, \eqref{eq:uniform_r_np}, and \eqref{eq:uniform_kp_np}. The case $r=2$ is straightforward as it is the largest size between \eqref{eq:uniform_k_nm} and \eqref{eq:uniform_r_np}. Assume now that $r \geq 3$. Let $n'_{1}=n-g, n'_{2}=n-k+r-\left\lceil \frac{k}{r} \right\rceil +1,$ and $n'_{3}(k')=n-k+k'-\max\{j,0\}$ where $2 \leq k' \leq r-1$ and $j=\left\lfloor\frac{-h}{k'}\right\rfloor +g$. 

First notice that both $n'_{2}$ and $n'_{3}(k')$ are upper bounded by $n-k+r-1$ since $\left\lceil \frac{k}{r} \right\rceil \geq 2$. Thus, if $g \leq k-r+1$, then $n'_{1} \geq n'_{2}$ and $n'_{1}\geq n'_{3}(k')$.

Suppose now that $g>k-r+1$ and let $k'=r-1$. Then, we have
\[
j =\left\lfloor\frac{k-(r-k')g}{k'}\right\rfloor = \left\lfloor\frac{k-g}{r-1}\right\rfloor < \left\lfloor\frac{r-1}{r-1}\right\rfloor = 1.
\]
Hence, $j\leq 0$ and $n'_{3}(r-1)=n-k+r-1$. We also have that $n'_{3}(k')< n-k+r-1$ for all $2 \leq k'\leq r-1$. Since we already saw that $n'_{2}\leq n-k+r-1$ and by the assumption on $g$, we have that $n'_{1} < n-k+r-1$, this implies that $n'_{3}(r-1)=n-k+r-1$ is the maximum size when $g>k-r+1$. 
\end{proof}

Figure \ref{fig:comparison_graph} displays the comparison between the length \eqref{eq:uniform_k_nm}, \eqref{eq:uniform_r_np}, and \eqref{eq:uniform_kp_np} for fixed $k$ and $r$. As we can see, \eqref{eq:uniform_k_nm} is the largest length in the high-rate regime while \eqref{eq:uniform_kp_np} is the largest length in the low-rate regime. While high-rate codes are preferable for storage, low-rate codes have advantages in terms of availability of hot data and lead to better rates when considering \emph{private information retrieval} schemes.

\begin{figure}
\centering
\includegraphics[height=5.9cm]{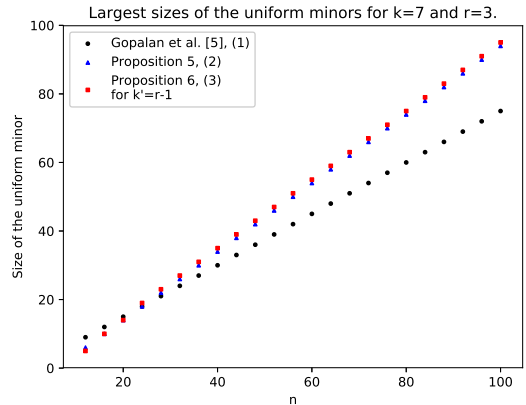}
\caption{Comparison between the sizes \eqref{eq:uniform_k_nm}, \eqref{eq:uniform_r_np}, and \eqref{eq:uniform_kp_np} when $n$ grows, $k=7$, and $r=3$.}
\label{fig:comparison_graph}
\end{figure}

Without assuming the MDS conjecture, we can use Lemma~\ref{lemma:mds_lowerbound} to obtain a lower bound on the field size. When applying the lemma to the uniform minor obtained in Proposition \ref{prop:uniform_kp_np}, the dimension cancels out and the bound is maximized when $k'=2$. 

\begin{theorem}
Any $(n,k,r)$ MR-LRC over $\F_{q}$ satisfies
\[
\left\lbrace
\begin{array}{ll}
q \geq n-k - \left\lceil \frac{k}{r} \right\rceil +2 & \text{if } r=2, \\
q \geq n-k+1 - \max\{j,0\} & \text{if } r \geq 3, \\ 
\end{array}
\right.
\]
where $j=\left\lfloor \frac{-h}{2} \right\rfloor +g$. 
\end{theorem} 

Even if these bounds are still far from the asymptotic bound in \cite{gopi19}, they improve the non-asymptotic bound in \cite{gopalan14}, which is $q\geq k+1$, for low-rate MR-LRCs. Indeed, a necessary condition for the new bounds to be better is that $n-k\geq k$. More precisely, the new bound for $r=2$ improves on $k+1$ when $\frac{k}{n} \leq \frac{2}{5}$. The bound when $r\geq 3$ improves on $k+1$ when $\frac{k}{n} \leq \frac{9}{20}$ for $r=3$ ; $\frac{k}{n} \leq \frac{12}{25}$ for $r=4$ ; and directly when $\frac{k}{n} \leq \frac{1}{2}$ for all $r\geq 5$.

\section{Conclusion}

In this letter, we studied maximally recoverable codes with a focus on classifying their uniform minors. As a direct consequence, we obtained the largest length of an MDS code inside an MR-LRC. Using the relation between MDS codes and the field size, we derived a lower bound on the required field size of an MR-LRC improving on the non-asymptotic bound in the low-rate regime. However, the gap between the lower bounds and the constructions remains an intriguing open problem. In particular, our results show that new techniques not relying on the MDS conjecture need to be found in order to close it.

\bibliographystyle{IEEEtran}
\bibliography{IEEEabrv,references_MRC}

\ifCLASSOPTIONcaptionsoff
  \newpage
\fi

\end{document}